\documentclass[12pt]{iopart}

\usepackage{iopams}
\usepackage[USenglish]{babel} 
\usepackage{dsfont}
\usepackage{rotating}
\usepackage{xcolor}
\usepackage{amsmath}
\usepackage{amssymb}
\usepackage{amsthm}
\usepackage[urlcolor=blue,colorlinks=true,linkcolor=blue,pdfstartview={FitH},bookmarks=false]{hyperref} 
\usepackage[square, sort, sort&compress, numbers, mcite]{natbib}
\usepackage{mciteplus} 
\newcommand{\ket}[1]{\ensuremath{|#1\rangle}}

\newcommand{\braket}[2]{\ensuremath{\langle#1|#2\rangle}}

\newcommand{\eg}{\emph{e.g.}}
\newcommand{\ie}{\emph{i.e.}}
\newcommand{\etal}{\emph{et al}}

\newcommand{\p}{\scriptscriptstyle{+}}
\newcommand{\m}{\scriptscriptstyle{-}}

\newcommand{\dg}{\dagger}

\newcommand{\mc}{\mathcal}

\newcommand{\txt}[1]{\text{#1}}

\newtheorem{thh}{Theorem}
\newtheorem{deff}{Definition}

\newtheorem{lem}{Lemma}

\mciteSetMidEndSepPunct{\newline\hspace*{0.0cm}}{.}{\relax}

\begin{document}

\title{New symmetry in the Rabi model}

\author{Bart{\l}omiej Gardas}
\address{Institute of Physics, University of Silesia, PL-40-007 Katowice, Poland}
\ead{bartek.gardas@gmail.com}

\author{Jerzy Dajka}
\address{Institute of Physics, University of Silesia, PL-40-007 Katowice, Poland}
\ead{jerzy.dajka@us.edu.pl}

\begin{abstract}
It is recognised that, apart from the total energy conservation, there is a nonlocal 
$\mathbb{Z}_2$ and a somewhat hidden symmetry in this model. Conditions for the existence
of this observable, its form, and its explicit construction are presented.
\end{abstract}

\pacs{03.65.Yz, 03.67.-a}
\maketitle

\section{Introduction}
\label{one}
A symmetry can be seen as an equivalence of different physical situations~\cite{peres}. Such an equivalence 
in quantum theory entails the invariance of a certain set of observables and can be formalised in terms of 
commutation relations between these operators and a given Hamiltonian. The existence of good quantum numbers,
also those having no classical counterpart, is a direct consequence of symmetries. It extends the amount of
information accessible for researchers studying quantum systems.

It is possible to explain an unusual system's behaviour, its properties and dynamics by means of symmetries.
Selection rules or Kramers degeneracy~\cite{kd} may serve as a good example here. Symmetries not only deepen
our understanding of quantum systems but also can be included to engineer their physical realisation more
effectively~\cite{fotorabi}. In general, the more symmetries recognised (together with related conserved
quantities) the more different approaches to study the system's dynamics are at our disposal.
 
Only in extreme cases, one can meet analytically solvable models (such as harmonic oscillator, Jaynes--Cummings
model or hydrogen atom) where symmetries can be found easily. In this paper, we consider a quantum model consisting
of a two--level system (qubit) interacting with a single mode bosonic field (electromagnetic radiation) with 
frequency $\omega$. The Hamiltonian of that system is assumed to be of the following form
\begin{equation}
\label{rabi}
{\bf{H}}=\beta\sigma_z+\Delta\sigma_x + \omega a^{\dagger}a+\sigma_z\otimes\left(g^*a+ga^{\dagger}\right),
\end{equation}
where $a$ and $a^{\dagger}$ are the creation and annihilation operators of the bosonic field. Mathematically, this means that
$[a,a^{\dagger}]=\mathbb{I}$. For an experimental characterisation of these operators see~\cite{kumar}. $\sigma_z$ and 
$\sigma_x$ denote the two Pauli spin matrices. The term $\beta\sigma_z$ stands for the unperturbed energy of the
qubit with possible eigenenergies $\pm\beta$. Tunnelling between the corresponding energy levels in the absence of the
bosonic field (spontaneous transition) is described by $\Delta\sigma_x$. Finally, the coupling constant $g$ reflects 
the strength of the interaction between the systems. 

The above Hamiltonian is the well--known Rabi model~\cite{rabiorg1,*rabiorg2}--probably the most influential
model describing fully quantized interaction between matter and light. Although the model originates from
quantum optics~\cite{vedral}, its applications range from molecular physics~\cite{molec}, solid state (see
Refs. in~\cite{irish}) to the recent experiments involving cavity and circuit QED~\cite{qed1,*qed2}. The Rabi 
model can be implemented by means of rich variety of different setups such as Josephson junctions~\cite{jj},
trapped ions~\cite{tj}, superconductors~\cite{super} or semiconductors~\cite{semi}, to name a few.

Despite its simplicity, the Hamiltonian of the Rabi model cannot be diagonalized exactly when $\Delta\not=0$. 
Although some progress has been reported recently~\cite{braak,*zieg}, exact analytical formulas for the eigenvalues 
and corresponding eigenfunction of the Hamiltonian~(\ref{rabi}) are still missing. There is a wide spectrum of 
available approximation techniques including rotating wave approximation~\cite{vedral} (leading to the famous
Jaynes--Cummings model~\cite{jc}) which allow the eigenproblem to be approached from many different directions.
At this point, a question concerning the existence of symmetries in the Rabi model (together with related constants
of motion) arises naturally. 

Provided that $\beta=0$, the Hamiltonian~(\ref{rabi}) remains unchanged when $\sigma_z\rightarrow-\sigma_z$ and 
$a\rightarrow-a$ (hence $a^{\dagger}\rightarrow-a^{\dagger}$). The symmetry operator $\bf{J}_0$ that generates 
this transformation (\eg{} fulfills $[{\bf{H}},{\bf{J}}_0]=0$) reads ${\bf{J}_0}=\sigma_x\otimes\txt{P}$, where
$\txt{P}=\exp(\txt{i}\pi a^{\dagger}a)$ is the bosonic parity~\cite{gardas4}. This is the well--known result:
still being unsolvable, the Rabi model possesses a discrete symmetry if $\beta=0$. 

When $\beta\not=0$, on the other hand, we can still leave $\bf{H}$ unaffected after changing $\sigma_z\rightarrow-\sigma_z$, 
$a\rightarrow-a$ if we change the sign of $\beta$ as well (\ie{} $\beta\rightarrow-\beta$). This instantly raises
a question: What does the corresponding generator of such transformation, $\bf{J}$, look like? Unfortunately, this
question has not been answered so far. Moreover, it was quite recently conjectured~\cite{braak} that the Rabi model does 
not possess any symmetry at all (except the trivial one related to the total energy conservation) as long as $\beta\not=0$. 
If that were true, the only self--adjoint operator $\bf{J}$ such that ${[\bf{H},\bf{J}]}=0$ would be the Hamiltonian 
$\bf{H}$ itself. 

On the basis of the results reported here, we prove that this conjecture is false. In particular, we show how one can 
find a self--adjoint involution $\bf{J}$, that is $\bf{J}^2=\mathbb{I}_{\txt{B}}$, such that $\bf{HJ}=\bf{JH}$. Also, 
we discuss the possibility of the exact diagonalization of the Rabi Hamiltonian~(\ref{rabi}).

It is worth mentioning that symmetry groups of the time evolution generator (the Hamiltonian $\bf{H}$ in our case) are larger 
than those of the corresponding equation of motion (Schr\"{o}dinger equation: $\ket{\dot{\Psi}_t}={\bf{H}}\ket{\Psi_t}$). 
In particular, we could consider the existence of a symmetry ${\bf{J}}_t$ which does not necessarily commute with $\bf{H}$ but
still assures the same time evolution for two different states: $\ket{\Psi_t}$, ${\bf{J}}_t\ket{\Psi_t}$. Of course, this
is possible if $i{\dot{\bf{J}}_t}=[{\bf{H}},{\bf{J}}_t]$. The idea of such dynamical symmetries is interesting by itself, 
yet it is beyond the scope of our current considerations and won't be pursued any further in this work.

\section{Main result}
\label{two}
Let us begin with formal rewriting of the Rabi Hamiltonian~(\ref{rabi}) as a matrix with operator entries:
\begin{equation}
 \label{bom}
 {\bf{H}} = 
  \begin{bmatrix}
  \txt{H}_{\p} & \Delta \\
   \Delta & \txt{H}_{\m} 
 \end{bmatrix},
 \quad\txt{where}\quad
 \txt{H}_{\pm}:=\omega a^{\dagger}a\pm\left(g^*a+ga^{\dagger}\right)\pm\beta.
\end{equation}
Customarily, the parameters $\Delta$ and $\beta$ denote $\Delta\mathbb{I}_{\txt{B}}$ and $\beta\mathbb{I}_{\txt{B}}$ 
respectively. $\mathbb{I}_{\txt{B}}$ stands for the identity on the bosonic Hilbert space $\mathcal{H}_{\txt{B}}$. 

The matrix representation of the Rabi model given in~(\ref{bom}) is established via a natural isomorphism 
$\mathbb{C}^2\otimes\mathcal{H}_{\txt{B}}\sim\mathcal{H}_{\txt{B}}\oplus\mathcal{H}_{\txt{B}}$. Usually, such
an identification is invoked in order to simplify purely algebraic calculations (see \eg{} ~\cite{simple}). 
This is not the reason why we use this idea here. Instead, we are going to attack the problem in question by 
using a concept of block operator matrix~\cite{bom} in conjunction with its relation to an operator Riccati 
equation~\cite{ricc_book,*gardas,*gardas2,*gardasNote}. 

First however, we would like to clarify some technical aspects concerning the Rabi matrix~(\ref{bom}) (\eg{}
its domain $\mathcal{D}(\bf{H})$). One should mention that this is not a primary issue in many papers addressing 
physical aspects of the Rabi model. Needless to say, one cannot take the advantage of very powerful
existential mathematical theorems (\eg{} the famous Banach fixed point theorem~\cite{simon}) in such cases 
simply because it is not known whether the premises of these statements are met.

In a first step toward constructing $\bf{J}$, we define domains $\mathcal{D}_{\pm}:=\mathcal{D}(\txt{H}_{\pm})$
on which both operators $\txt{H}_{\pm}$ are self--adjoint. Since the off--diagonal elements of $\bf{H}$ are bounded, 
we have $\bf{H}^{*}=\bf{H}$ on $\mathcal{D}({\bf{H}})=\mathcal{D}_{\p}\oplus\mathcal{D}_{\m}$. As both $a$ and 
$a^{\dagger}$ are unbounded, the canonical commutation relation holds only on some (dense) subspace 
$\mc{D}_2$ of $\mc{H}_{\text{B}}$. Let us assume that $\mc{D}_1$ is a dense set on which $a$ and $a^{\dagger}$ 
are adjoint to each other \ie, $(a^{\dagger})^{*}=a$ and $a^{*}=a^{\dagger}$. At this point, it is not obvious
that the subspaces having the desired properties exist at all. An interested reader can find the detailed construction
of $\mc{D}_i$ \eg{} in~\cite{berezin,*br,*sz0}. Here, we briefly summarise what was covered therein. We have

\begin{equation}
\label{dom}
 \mc{D}_i = \left\{\sum_{n=0}^{\infty}\xi_n\ket{n}\in\mathcal{H}_{\txt{B}}:\quad\sum_{n=0}^{\infty}n^i|\xi_n|^2<\infty\right\},
 \quad i =1,2,
\end{equation}
where $\{\ket{n}\}_{n=0}^{\infty}$ is the canonical (orthonormal) basis in $l_2$ ($\sim\mc{H}_{\txt{B}}$). 
Considering the fact that $a$, $a^{\dagger}$ and $a^{\dagger}a$ ought to produce normalizable states, the 
above definitions seem natural. Having~(\ref{dom}) in place, we define 

\begin{equation}
 \label{aa}
   a\ket{\psi} := \sum_{n=1}^{\infty}\sqrt{n}\xi_n\ket{n-1}, \quad
   a^{\dagger}\ket{\psi} := \sum_{n=0}^{\infty}\sqrt{n+1}\xi_n\ket{n+1},
   \quad\ket{\psi}\in\mc{D}_1.
  \end{equation}
It follows immediately from~(\ref{aa}) that $a^{\dagger}\ket{n}=\sqrt{n+1}\ket{n+1}$ and $a\ket{n}=\sqrt{n}\ket{n-1}$.
Interestingly, the latter relations serve as the very definition of the creation and annihilation operators in most 
textbooks on quantum mechanics. A definition like this may be well motivated physically, yet it has at least one 
serious mathematical drawback. Namely, it introduces closeable operators which are not closed. This leads to a 
variety of technical difficulties typical for such classes of operators. One can avoid them by taking the 
closures~(\ref{aa}) as proper definitions, instead. 

A basic result from operator theory (see \eg, Theorem 4.2.7 in~\cite{blank}) states that if $\txt{A}$ is closed on 
$\mathcal{D}(\txt{A})$ then $\txt{A}^*\txt{A}$ is positive, self-adjoint and its domain is a core of $\txt{A}$ (\ie{} 
$\txt{A}$ is the closure of its restriction $\txt{A}_{|\mathcal{D}(\txt{A}^*\txt{A})}$). On $\mathcal{D}_2$, the operators 
$\txt{H}_{\pm}$ can be written as 
\begin{equation}
\label{compact}
\txt{H}_{\pm}=\omega\left(a\pm\frac{g}{\omega}\right)^{\dg}\left(a\pm\frac{g}{\omega}\right)\pm\beta-\frac{|g|^2}{\omega},
\end{equation}
and as a result, they are both self--adjoint and their common domain $\mathcal{D}_2$ is a core of both $a$ and
$a^{\dagger}$. In conclusion, the Rabi Hamiltonian~(\ref{bom}) is well defined and self--adjoint on 
$\mathcal{D({\bf{H}})}=\mathcal{D}_2\oplus\mathcal{D}_2$. 

After discussing technical nuances concerning the Rabi model, we introduce a quadratic second order operator equation,
known as the Riccati equation, which has the following form
\begin{equation}
\label{ricc}
 \Delta\txt{X}^2+\txt{X}\txt{H}_{\p}-\txt{H}_{\m}\txt{X}-\Delta=0.
\end{equation}
Many of the relevant problems related to the Rabi model~(\ref{rabi}), including its exact diagonalization, can be reduced 
to the mathematical questions concerning solvability of this equation.

There is more than one notion of a solution when equations with operator coefficients are involved. In Hilbert spaces, one 
can define a solution in terms of the scalar product (weak solution). On the other hand, one may require for operators to 
be equal when they produce the same results while acting on the same states. These kind of solutions, which are of great 
importance in quantum mechanics, are known as strong ones. Let us briefly clarify these two notions for the Riccati equation
in question. 

\begin{deff}
A bounded operator $\txt{X}_0$ acting on a Hilbert space $\mathcal{H}$ is called a weak solution of the Riccati 
equation~(\ref{ricc}) if
\begin{equation}
\label{weak}
\Delta\langle\txt{X}_0^2\phi,\psi\rangle + \langle\txt{X}_0\txt{H}_{\p}\phi,\psi\rangle
                                         -\langle\txt{X}_0\phi,\txt{H}_{\m}\psi\rangle-\Delta\langle\phi,\psi\rangle=0,
\quad\txt{for}\quad                           
\ket{\psi}, \ket{\phi}\in\mathcal{D}_2.              
\end{equation}
A bounded operator $\txt{X}_0$ acting on $\mathcal{H}$ such that $Ran(\txt{X}_{0|\mathcal{D}_2})\subset\mathcal{D}_2$
and
\begin{equation}
\label{strong}
\Delta\txt{X}_0^2\ket{\psi}+\txt{X}_0\txt{H}_{\p}\ket{\psi}-\txt{H}_{\m}\txt{X}_0\ket{\psi}-\Delta\ket{\psi}=0,
 \quad\txt{for}\quad\ket{\psi}\in\mathcal{D}_2,
\end{equation}
is a strong solution of~(\ref{ricc}). 
\end{deff}
Of course, a strong solution is also a weak solution. It is often easier to prove the existence of a weak rather than
a strong solution. However, strong solutions, especially in quantum mechanics, are those which we are interested in.
Fortunately, the two notions are in fact equivalent~\cite{strong}. Nevertheless, there is no general method of finding 
either weak or strong solutions to a particular Riccati equation. For this reason, the following theorem, which provides 
criteria of solvability, is of great importance to us.
\begin{lem}
\label{mt}
Let $\txt{H}_{\pm}$ be (possibly unbounded) self--adjoint operators acting on domains $\mathcal{D}(\txt{H}_{\pm})$
in a separable Hilbert space $\mathcal{H}$. Let us also assume that $\txt{V}_1\not=0$ and $\txt{V}_2$ are bounded 
operators on $\mathcal{H}$. If the spectra $\sigma(\txt{H}_{\pm})$ are disjoint, i.e., 
\begin{equation}
\label{d}
 d:=\txt{dist}\left(\sigma(\txt{H}_{\p}),\sigma(\txt{H}_{\m})\right)>0,
\end{equation}
and if $\txt{V}_1$, $\txt{V}_2$ satisfy the `smallness assumption'
\begin{equation}
\label{small}
 \sqrt{\|\txt{V}_1\|\|\txt{V}_2\|}<\frac{d}{\pi},
\end{equation}
then the Riccati equation
\begin{equation}
 \txt{X}\txt{V}_1\txt{X}+\txt{X}\txt{H}_{\p}-\txt{H}_{\m}\txt{X}-\txt{V}_2=0,
\end{equation}
has a unique weak solution $\txt{X}_0$ in the ball
\begin{equation}
\left\{\txt{X}\in\mathcal{B}(\mathcal{H}): \|\txt{X}\|<\frac{d}{\pi\|\txt{V}_1\|} \right\}.
\end{equation}
satisfying an estimate
\begin{equation}
\|\txt{X}_0\|\leq\frac{1}{\|V_2\|}\left(\frac{d}{\pi}-\sqrt{\frac{d^2}{\pi^2}-\|\txt{V}_1\|\|\txt{V}_2}\|\right).
\end{equation}
In particular, if
\begin{equation}
\label{cont}
\|\txt{V}_1\|+\|\txt{V}_2\|<\frac{2d}{\pi},
\end{equation}
then $\txt{X}_0$ is a strict contraction, that is, $\|X_0\|<1$.
\end{lem}
An elegant and compact proof of this statement, based on the Banach fixed point theorem, can be found in~\cite{main}. 

Now, let us prove our main result. First, we show that the existence of a solution of the Riccati equation~(\ref{ricc}) implies
the existence of an operator generating a symmetry in the system~(\ref{bom}). Second, we argue that under certain conditions 
imposed on the parameters $\Delta$, $\beta$, $\omega$ this equation is weakly solvable.
\begin{thh}
\label{symm}
Let us assume that there exists a weak (and hence strong) solution $\txt{X}_0$ of the Riccati equation~(\ref{ricc}).
Then there also exists a self--adjoint involution $\bf{J}$ such that $\bf{JH}=\bf{HJ}$ where $\bf{H}$ is given 
by~(\ref{bom}). Moreover, the generator $\bf{J}$ in terms of $\txt{X}_0$ reads
\begin{equation}
\label{j}
{\bf{J}}=
\begin{bmatrix}
 \txt{J}_0-1 & \txt{J}_0\txt{X}_0^* \\
\txt{X}_0\txt{J}_0 & \txt{X}_0\txt{J}_0\txt{X}_0^*-1
\end{bmatrix},
\quad\txt{where}\quad
\txt{J}_0=2(1+\txt{X}_0^*\txt{X}_0)^{-1}.
\end{equation}
\end{thh}
\begin{proof}
Let $\mathcal{G}(\txt{X}_0)$ be the graph of $\txt{X}_0$, that is
\begin{equation}
\label{graph}
\mathcal{G}(\txt{X}_0)=
\left\{
\begin{bmatrix}
\ket{\psi}\\ 
\txt{X}_0\ket{\psi}
\end{bmatrix}
\in\mathcal{H}_{\txt{B}}\oplus\mathcal{H}_{\txt{B}}
:\ket{\psi}\in\mathcal{H}_{\txt{B}}
\right\}.
\end{equation}
$\txt{X}_0$ is a strong solution of~(\ref{ricc}) thus $\txt{X}_0\ket{\psi}\in\mathcal{D}_2$
(by definition) and $\txt{X}_0(\txt{H}_{\p}+\Delta\txt{X}_0)\ket{\psi}=(\txt{H}_{-}\txt{X}_0+\Delta)\ket{\psi}$
for $\ket{\psi}\in\mathcal{D}_2$. Therefore,
\begin{equation}
 \label{include}
  \begin{bmatrix}
  \txt{H}_{\p} & \Delta \\
   \Delta & \txt{H}_{\m} 
 \end{bmatrix}
 \begin{bmatrix}
 \ket{\psi} \\
 \txt{X}_0\ket{\psi}
 \end{bmatrix}
 =
  \begin{bmatrix}
 \left(\txt{H}_{\p}+\Delta\txt{X}_0\right)\ket{\psi} \\
 \txt{X}_0\left(\txt{H}_{\p}+\Delta\txt{X}_0\right)\ket{\psi} 
 \end{bmatrix}
 \in\mathcal{G}(\txt{X}_0),
\end{equation}
that is ${\bf{H}}(\mathcal{G}(\txt{X}_0)\cap\mathcal{D}_2)\subset\mathcal{D}_2$. Making use of the 
same arguments, one can verify that $\mathcal{G}(\txt{X}_0)^{\bot}$, which is given by
\begin{equation}
\label{gort}
\mathcal{G}(\txt{X}_0)^{\bot}=
\left\{
\begin{bmatrix}
-\txt{X}_0^{*}\ket{\psi}\\ 
\ket{\psi}
\end{bmatrix}
\in\mathcal{H}_{\txt{B}}\oplus\mathcal{H}_{\txt{B}}
:\ket{\psi}\in\mathcal{H}_{\txt{B}}
\right\},
\end{equation}
is ${\bf{H}}$--invariant as well. $\txt{X}_0$ is bounded and thus its graph forms a closed subspace of
$\mathcal{H}_{\txt{B}}\oplus\mathcal{H}_{\txt{B}}$ and hence the decomposition 
$\mathcal{H}_{\txt{B}}\oplus\mathcal{H}_{\txt{B}}=\mathcal{G}(\txt{X}_0)\oplus\mathcal{G}(\txt{X}_0)^{\bot}$
holds true. Therefore, each state $\ket{\Psi}\in\mathcal{D}(\bf{H})$ of the composite system can be uniquely decomposed
$\ket{\Psi}=\ket{\Psi_1}\oplus\ket{\Psi_2}$ where $\ket{\Psi_1}\in\mathcal{G}(\txt{X}_0)$ and $\braket{\Psi_2}{\Psi_1}=0$.

Let ${\bf{P}}_{\p}$ be a projection onto $\mathcal{G}(\txt{X}_0)$. Then it follows that ${\bf{P}}_{\p}{\bf{H}}\ket{\Psi_1}={\bf{H}}\ket{\Psi_1}$ 
and ${\bf{P}}_{\p}{\bf{H}}\ket{\Psi_2}=0$. Assuming for a moment that ${\bf{P}}_{\p}\mathcal{D}_2\subset\mathcal{D}_2$, we obtain
\begin{equation}
{\bf{H}}\left({\bf{P}}_{\p}\ket{\Psi_1}\oplus{\bf{P}}_{\p}\ket{\Psi_2}\right)= {\bf{H}}\ket{\Psi_1}
\quad\txt{and}\quad
{\bf{P}}_{\p}\left({\bf{H}}\ket{\Psi_1}\oplus{\bf{H}}\ket{\Psi_2}\right)= {\bf{H}}\ket{\Psi_1}.
\end{equation}
Therefore, ${\bf{H}}{\bf{P}}_{\p}\ket{\Psi}={\bf{P}}_{\p}{\bf{H}}\ket{\Psi}$ for all $\ket{\Psi}\in\mathcal{D}_2$.

The inverse $(1+\txt{X}_0^*\txt{X}_0)^{-1}$ exists and it is a bounded self--adjoint operator on $\mathcal{H}_{\txt{B}}$.
Thus, $\bf{P}_{\p}$ can be expressed as
\begin{equation}
\label{pp}
{\bf{P}_{\p}}=
\frac{1}{2}
\begin{bmatrix}
 \txt{J}_0 & \txt{J}_0\txt{X}_0^* \\
\txt{X}_0\txt{J}_0 & \txt{X}_0\txt{J}_0\txt{X}_0^*
\end{bmatrix}.
\end{equation} 
It is a matter of straightforward calculations to see that~(\ref{pp}) indeed projects onto $\mathcal{G}(\txt{X}_0)$.

Due to the fact that ${\bf{J}}=2{\bf{P}_{\p}}-{\bf{1}}$, the only question which we need to address to conclude the 
proof is whether ${\bf{P}_{\p}}\ket{\Psi}$ is again in $\mathcal{D}(\bf{H})$ for $\ket{\Psi}\in\mathcal{D}(\bf{H})$. 
Because $\txt{X}_0$ is a weak (and hence strong) solution of~(\ref{ricc}), we have $\txt{X}_0\mathcal{D}_2\subset\mathcal{D}_2$. 
Moreover, the function $f(\psi):=\langle\txt{H}_{\p}\psi,\txt{X}_0^*\phi\rangle$ is continuous on $\mathcal{D}_2$
for every $\ket{\phi}\in\mathcal{D}_2$. Indeed, it follows from~(\ref{weak}) that
\begin{equation}
|f(\psi)|\leq M_{\phi}\|\psi\|,
\quad\txt{where}\quad 
M_{\phi}=\alpha\|\phi\|\|\txt{X}_0\|^2+\|\txt{H}_{\m}\phi\|\|\txt{X}_0\|+\alpha\|\phi\|.
\end{equation}
As a result, $\txt{X}_0^*\ket{\phi}\in\mathcal{D}(\txt{H}_{\p}^*)=\mathcal{D}_2$, \ie{} $\txt{X}_0^*\mathcal{D}_2\subset\mathcal{D}_2$
and therefore $\txt{J}_0^{-1}\mathcal{D}_2\subset\mathcal{D}_2$.  $\txt{J}_0^{-1}$ is invertible, hence $\txt{J}_0\mathcal{D}_2=\mathcal{D}_2$. 
In summary, ${\bf{P}_{\p}}\mathcal{D}({\bf{H}})\subset\mathcal{D}({\bf{H}})$ which concludes the proof.\qedhere
\end{proof}

\begin{thh}
Let us assume that $\beta$, $\omega$, $\Delta\not=0$ satisfy the following conditions 
\begin{equation}
\label{ass}
 \frac{2\beta}{\omega}\notin\mathbb{N}\quad\txt{and}\quad \frac{\Delta}{\beta}>\frac{\pi}{2}. 
\end{equation}
Then there exists a unique weak (hence strong) solution of the Riccati equation~(\ref{ricc}) such that $\|X_0\|<1$.
As a result, there is a $\mathbb{Z}_2$ symmetry with respect to which the Rabi model is invariant. The generator of 
this symmetry is given by~(\ref{symm}).
\end{thh}
\begin{proof}
$\txt{V}_x=\txt{i}(xa^{\dagger}-x^*a)$ is self--adjoint for $x\in\mathbb{C}$ thus the unitary Weyl operator 
$\txt{D}_x=\exp(\txt{i}\txt{V}_x)$ is well defined. Moreover, $\txt{D}_{x}^*=\txt{D}_{-x}$ and therefore
\begin{equation}
\txt{H}_{\pm}=\txt{D}_{\pm\frac{g}{\omega}}\left(\omega a^{\dagger}a\pm\beta-\frac{|g|^2}{\omega}\right)\txt{D}_{\mp\frac{g}{\omega}}.
\end{equation} 
By virtue of $a^{\dagger}a\ket{n}=n\ket{n}$ (keep in mind that $n\in\mathbb{N}$), we have
\begin{equation}
\label{spec}
\sigma(\txt{H}_{\pm})=\left\{\omega n\pm\beta-\frac{|g|^2}{\omega}:n\in\mathbb{N}\right\}
=\omega\mathbb{N}\cup\{\pm\beta\}-\frac{|g|^2}{\omega}.
\end{equation}
If $2\beta$ is not multiple of $\omega$ then the distance 
\begin{equation}
\txt{dist}(\sigma(\txt{H}_{\p}),\sigma(\txt{H}_{\m}))=
\txt{inf}\{|\omega(n-m)+2\beta|: n,m\in\mathbb{N}\}=2\beta\not=0.
\end{equation}
Therefore, the spectra $\sigma(\txt{H}_{\pm})$ are disjoint \ie, the condition~(\ref{d}) holds true. In addition, both the 
smallness assumption~(\ref{small}) and~(\ref{cont}) imposed on the off--diagonal elements are satisfied as long as 
$2\Delta>\pi\beta$. According to Lemma~\ref{mt}, there is exactly one solution of the Riccati equation~(\ref{ricc}) 
and it is a strict contraction ($\|\txt{X}_0\|<1$).

 The second statement of the theorem follows immediately from Theorem~\ref{symm}.\qedhere
\end{proof}

\section{Discussion}
We begin with the $\beta=0$ case in which the spectra~(\ref{spec}) overlap and hence the separability condition~(\ref{d}) is not
satisfied. Therefore, one cannot invoke Lemma~\ref{mt} to establish the existence of a solution to the Riccati equation~(\ref{ricc}). 
However, the spectra $\sigma(\txt{H}_{\pm})$ in that particular case are identical and $\txt{H}_{\pm}$ can be transformed one into 
another by the same bosonic parity operator that generates the symmetry $\bf{J}_0$. This is not an accidental coincidence as $\txt{P}$ 
is a solution of the Riccati equation~(\ref{ricc}). Indeed, 
\begin{equation}
\label{parity}
\txt{P}\ket{\psi}=\sum_{n=0}^{\infty}(-1)^n\xi_n\ket{n},
\quad\txt{where}\quad \xi_n=\braket{n}{\psi},
\end{equation}
%
from which it follows immediately that $\txt{P}$ is bounded and $\txt{Ran}(\txt{P}_{|\mathcal{D}_2})\subset\mathcal{D}_2$. 
Note, if $n\xi_n$ are square--summable, $\sum_n |n\xi_n|^2<\infty$, so are $(-1)^{n}n\xi_n$. In the light of~(\ref{aa}),
we obtain $\txt{P}a\txt{P}=-a$ as well as $\txt{P}a^{\dagger}\txt{P}=-a^{\dagger}$ and finally  
$\txt{P}\txt{H}_{\pm}\txt{P}=\txt{H}_{\mp}$. And because $\txt{P}$ is a self--adjoint involution, it solves the Riccati 
equation~(\ref{ricc}) as stated. 

At this point, we would like to make some remarks. First and foremost, $\txt{P}$ is not a unique solution of the Riccati 
equation~(\ref{ricc}). For instance, $-\txt{P}$ also satisfies this equation. Second, the symmetry generator ${\bf{J}}$ 
from Theorem~\ref{symm} reads $\pm{\bf{J}}_0$ when $\txt{X}_0=\pm\txt{P}$ as one may expect.

If the conditions~(\ref{ass}) are met, in particular for $\beta\not=0$, the spectra $\txt{H}_{\pm}$ are separated and the
Riccati equation~(\ref{ricc}) possesses exactly one solution $\txt{X}_0$. According to Theorem~\ref{symm}, this solution 
corresponds to a symmetry generator $\bf{J}$. The only problem is that $\txt{X}_0$ is unknown. One can attempt to simplify
the problem by putting $\txt{X}_0=\txt{Y}_{\beta}\txt{P}$, where
\begin{equation}
\label{nricc}
\alpha\txt{Y}_{\beta}\txt{P}\txt{Y}_{\beta}+\left[\txt{Y}_{\beta},\txt{H}_{\p}\right]+2\beta\txt{Y}_{\beta}-\alpha\txt{P}=0,
\end{equation}
and $\txt{H}_{\p}$ is redefined so that it reads~(\ref{compact}) for $\beta=0$. This equation becomes trivial and its 
solution reads $\txt{Y}_0=1$ when $\beta=0$. On the other hand, as long as $\beta\neq 0$, under~(\ref{ass}), the 
premises of Lemma~\ref{mt} are satisfied. Hence, a unique $\txt{Y}_{\beta}$ exists and $\|\txt{Y}_{\beta}\|\leq 1$.
Moreover, if the inverse $\txt{Y}_{\beta}^{-1}$ exists as well then
\begin{equation}
\alpha\txt{Y}_{\beta}^{-1}\txt{P}\txt{Y}_{\beta}^{-1}+\left[\txt{Y}_{\beta}^{-1},\txt{H}_{\p}\right]
+2(-\beta)\txt{Y}_{\beta}^{-1}-\alpha\txt{P}=0,
\end{equation}
and therefore $\txt{Y}_{-\beta}=\txt{Y}_{\beta}^{-1}$. Although we cannot solve~(\ref{nricc}) either, the latter equality
indicates the class which $\txt{Y}_{\beta}$ belongs to. One can also verify that the operator $\txt{Y}_{\beta}$ is not 
self--adjoint provided it is a function of $\txt{H}_{\p}$ and it cannot be anti--self--adjoint ($\txt{Y}_{\beta}^*=-\txt{Y}_{\beta}$) 

Indeed, if $\txt{H}_{\p}$ such that $\txt{Y}_{\beta}=\txt{Y}_{\beta}^*$ exists,~(\ref{nricc}) would imply the following 
separation into a self--adjoint and anti--self--adjoin part
\begin{equation}
\alpha\txt{Y}_{\beta}\txt{P}\txt{Y}_{\beta}+2\beta\txt{Y}_{\beta}-\alpha\txt{P}=0, 
\quad\txt{and}\quad
\left[\txt{Y}_{\beta},\txt{H}_{\p}\right]=0.
\end{equation}
Both these equations can be solved separately, but the solutions do not agree with each other unless $\beta=0$.
Similar arguments show that the condition $\txt{Y}_{\beta}^*=-\txt{Y}_{\beta}$ is necessary for $\txt{Y}_{\beta}=0$.
This contradicts~(\ref{nricc}) even when $\beta=0$.

Solutions of the Riccati equation~(\ref{ricc}) can also be used to obtain the eigenfunctions and corresponding 
eigenvalues of the Rabi Hamiltonian. Let us briefly discuss the idea. 

Both $\mathcal{G}(\txt{X}_0)$ and $\mathcal{G}(\txt{X}_0)^{\bot}$ are $\bf{H}$--invariant. Thus, if $\ket{\Psi}$ is 
an energy eigenstate then either $\ket{\Psi}\in\mathcal{G}(\txt{X}_0)$ or $\ket{\Psi}\in\mathcal{G}(\txt{X}_0)^{\bot}$.
Actually, we can say more than that. Let $\txt{Z}_{\p}=\txt{H}_{\p}+\Delta\txt{X}_0$ and
$\txt{Z}_{\m}=\txt{H}_{\m}-\Delta\txt{X}_0^*$ be defined on $\mathcal{D}_2$. Together with~(\ref{include}), this gives
\begin{equation}
\ket{\Psi_{\lambda}}
=\begin{bmatrix}
\ket{\psi_{\lambda}} \\
\txt{X}_0\ket{\psi_{\lambda}}
\end{bmatrix},
\quad\txt{where}\quad
\txt{Z}_{\p}\ket{\psi_{\lambda}}=\lambda\ket{\psi_{\lambda}},
\end{equation}
provided $\ket{\Psi_{\lambda}}$ is in $\mathcal{G}(\txt{X}_0)$.

Also, one can verify that all eigenstates from $\mathcal{G}(\txt{X}_0)^{\bot}$ are of the form:
\begin{equation}
\ket{\Phi_{\lambda}}
=\begin{bmatrix}
-\txt{X}_0^*\ket{\phi_{\lambda}} \\
\ket{\phi_{\lambda}}
\end{bmatrix},
\quad\txt{where}\quad
\txt{Z}_{\m}\ket{\phi_{\lambda}}=\lambda\ket{\phi_{\lambda}}.
\end{equation}
It can be proven that $\txt{Z}_{\pm}$ are self--adjoint on Hilbert spaces
$(\mathcal{H}_{\txt{B}},\langle(1+\txt{X}_0^*\txt{X}_0)\cdot,\cdot\rangle)$ and
$(\mathcal{H}_{\txt{B}},\langle(1+\txt{X}_0\txt{X}_0^*)\cdot,\cdot\rangle)$, respectively~\cite{bom}.
Moreover, $\sigma({\bf{H}})=\sigma(\txt{Z}_{\p})\cup\sigma(\txt{Z}_{\m})$ and the following similarity
relation holds
\begin{equation}
\label{pch}
{\bf{S}}^{-1}
\begin{bmatrix}
\txt{H}_{\p} & \Delta \\
\Delta & \txt{H}_{\m}
\end{bmatrix}
{\bf{S}}
=
\begin{bmatrix}
\txt{H}_{\p}+\Delta\txt{X}_0 & 0 \\
0 & \txt{H}_{\m}-\Delta\txt{X}^*_0
\end{bmatrix},
\quad\txt{where}\quad
{\bf{S}}=
\begin{bmatrix}
1 & -\txt{X}^*_0 \\
\txt{X}_0 & 1
\end{bmatrix}.
\end{equation}
The above block diagonal form of $\bf{H}$ extends the notion of the parity chains introduced in~\cite{braak}.

\section{Summary}

We have recognised a symmetry of the Rabi Hamiltonian and constructed its generator ${\bf{J}}$. Although this symmetry is
nonlocal (unlike \eg{} ${\bf{J}}_0=\sigma_z\otimes e^{\txt{i}\pi a^{\dagger}a}$), it is a self--adjoint involution.  
Therefore, it can be considered as a generalised parity of the Rabi model. Invoking physical nomenclature, the Rabi model 
is invariant with respect to this parity or it has an \emph{unbroken} $\mathbb{Z}_2$ symmetry. In literature, the latter
terminology is often used in a different (local) context where it is stated that the $\beta\not=0$ case corresponds to
a broken $\mathbb{Z}_2$ symmetry (because $[\bf{H},\bf{J}_0]\not=0$). Our aim was to generalise the local parity combined
by the parity operators of the individual subsystems: $\sigma_x$ and $e^{\txt{i}\pi a^{\dagger}a}$ to the nonlocal one for
$\beta\not=0$.  

Our results are not of purely existential character. By means of a solution to an operator Riccati type equation, we 
have derived an explicit formula for the generator $\bf{J}$ and formulated conditions (range of parameters~(\ref{ass}))
guaranteeing its existence. The question whether the generator $\bf{J}$ can exist under conditions other than~(\ref{ass})
remains open. This problem is a subject of our current intensive investigation.

At this point one should mention that usually the existence of a discrete symmetry in a quantum system is not enough by 
itself to fully understand its dynamics. Also, there is no obvious and direct guideline suggesting usefulness of symmetries
given by discrete operators, especially nonlocal ones, in construction of solutions to the equations of motion of composite
systems. However, discrete symmetries, local or not, allow the decomposition of the system Hilbert space into two subspaces with 
states having certain properties. One can then seek for the solution to the equation of motion in the individual subspaces
(and then try to combine the results to obtain a full solution). For the Rabi model, in the case of local parity, this idea
can be realised in terms of so called parity chains~\cite{braak}. The generalisation to the nonlocal case can by accomplished
by means of block diagonalization according to~(\ref{pch}). The latter formula may also serve as a good starting point for
developing new analytical approximations or numerical treatment of the eigenproblem~\cite{approx,*srwa}. 

Moreover, nonlocal discrete symmetries can help in classification and grouping of known solutions~\cite{braak}. They can 
also be used in constructing new solutions from the ones which are already known such as Juddian solutions~\cite{judd,*em} 
or so called quasi--exact solutions~\cite{qes}. Symmetries of the type presented here can also serve as a tool helping to 
verify certain conjectures concerning solutions of the Rabi model such as the celebrated Reik conjecture~\cite{reik}.  

We would like to emphasise that there is always a physical context (beyond mathematics) of studying symmetries (both local 
and nonlocal) in physics. For instance, there is a connection between symmetries of a quantum system and good quantum numbers
in that system~\cite{gqn}. Any measurement confirming conservation of such numbers confirms, at least partially, correctness
of the model (\ie, whether a given choice of the Hamiltonian properly describes the system). As `quantum phenomena do not 
occur in a Hilbert space, they occur in a laboratory'~\cite{peres}, the more symmetries to our disposal the more tests can
be performed. This ultimately verifies our understanding of quantum systems and their behaviour.

It seems that an inability to solve the Riccati equation when $\beta\not=0$ is the core reason why the symmetry~(\ref{j}) 
hasn't been recognised earlier. Although the solution of this equation exists as we have proved, it may not be expressible
by standard (well--known) operators. In that case, it is very unlikely to find the explicit form of ${\bf{J}}$ also by means
of different methods regardless of their nature. On the other hand, the Riccati equation can easily be solved in terms of 
the well-known bosonic parity when $\beta=0$. As one may expect, the corresponding generator ${\bf{J}}_0$ has been known all along.

The solvability problem of the Riccati equation can also be related to the question regarding diagonalization of the Rabi
model. In this paper, we have investigated the possibility of finding the eigenvalues and eigenvector of the Rabi Hamiltonian.
We have not offered full resolution, yet compact and exact expressions have been derived that, to some extent, simplify the 
problem. Although our analysis was mainly focused on the Rabi model, the presented scheme of diagonalization can be extended 
to general qubit--environment models.  


\newpage

\ack
This work was supported by the Polish Ministry of Science and Higher Education under project {\bf{Iuventus Plus}}, 
No. 0135/IP3/2011/71 (B. G) and NCN Grant N202 052940 (J. D)



\begin{mcitethebibliography}{10}
\expandafter\ifx\csname url\endcsname\relax
  \def\url#1{{\tt #1}}\fi
\expandafter\ifx\csname urlprefix\endcsname\relax\def\urlprefix{URL }\fi
\providecommand{\eprint}[2][]{\url{#2}}

\bibitem{peres}
Peres A 1993 {\em Quantum Theory: Concepts and Methods\/} (Kluwer Academic
  Publishers, London)\relax
\mciteBstWouldAddEndPuncttrue
\mciteSetBstMidEndSepPunct{\mcitedefaultmidpunct}
{\mcitedefaultendpunct}{\mcitedefaultseppunct}\relax
\EndOfBibitem
\bibitem{kd}
Roberts B~W 2012 {\em Phys. Rev.\/} A
  \href{http://dx.doi.org/10.1103/PhysRevA.86.034103}{{\bf 86}(3) 034103}\relax
\mciteBstWouldAddEndPuncttrue
\mciteSetBstMidEndSepPunct{\mcitedefaultmidpunct}
{\mcitedefaultendpunct}{\mcitedefaultseppunct}\relax
\EndOfBibitem
\bibitem{fotorabi}
Crespi A, Longhi S and Osellame R 2012 {\em Phys. Rev. Lett.\/}
  \href{http://dx.doi.org/10.1103/PhysRevLett.108.163601}{{\bf 108}(16)
  163601}\relax
\mciteBstWouldAddEndPuncttrue
\mciteSetBstMidEndSepPunct{\mcitedefaultmidpunct}
{\mcitedefaultendpunct}{\mcitedefaultseppunct}\relax
\EndOfBibitem
\bibitem{kumar}
Kumar R, Barrios E and Kupchak C~{\etal} 2013 {\em Phys. Rev. Lett.\/}
  \href{http://dx.doi.org/10.1103/PhysRevLett.110.130403}{{\bf 110}(13)
  130403}\relax
\mciteBstWouldAddEndPuncttrue
\mciteSetBstMidEndSepPunct{\mcitedefaultmidpunct}
{\mcitedefaultendpunct}{\mcitedefaultseppunct}\relax
\EndOfBibitem
\bibitem{rabiorg1}
Rabi I~I 1936 {\em Phys. Rev.\/}
  \href{http://dx.doi.org/10.1103/PhysRev.49.324}{{\bf 49}(4) 324--328}\relax
\mciteBstWouldAddEndPuncttrue
\mciteSetBstMidEndSepPunct{\mcitedefaultmidpunct}
{\mcitedefaultendpunct}{\mcitedefaultseppunct}\relax
\EndOfBibitem
\bibitem{rabiorg2}
Rabi I~I 1937 {\em Phys. Rev.\/}
  \href{http://dx.doi.org/10.1103/PhysRev.51.652}{{\bf 51}(8) 652--654}\relax
\mciteBstWouldAddEndPuncttrue
\mciteSetBstMidEndSepPunct{\mcitedefaultmidpunct}
{\mcitedefaultendpunct}{\mcitedefaultseppunct}\relax
\EndOfBibitem
\bibitem{vedral}
Vedral V 2006 {\em Modern Foundations of Quantum Optics\/} (Imperial College
  Press, London)\relax
\mciteBstWouldAddEndPuncttrue
\mciteSetBstMidEndSepPunct{\mcitedefaultmidpunct}
{\mcitedefaultendpunct}{\mcitedefaultseppunct}\relax
\EndOfBibitem
\bibitem{molec}
Thanopulos I, Paspalakis E and Kis Z 2004 {\em Chem. Phys. Lett.\/}
  \href{http://dx.doi.org/10.1016/j.cplett.2004.03.129}{{\bf 390} 228 --
  235}\relax
\mciteBstWouldAddEndPuncttrue
\mciteSetBstMidEndSepPunct{\mcitedefaultmidpunct}
{\mcitedefaultendpunct}{\mcitedefaultseppunct}\relax
\EndOfBibitem
\bibitem{irish}
Irish E~K 2007 {\em Phys. Rev. Lett.\/}
  \href{http://dx.doi.org/10.1103/PhysRevLett.99.173601}{{\bf 99}(17)
  173601}\relax
\mciteBstWouldAddEndPuncttrue
\mciteSetBstMidEndSepPunct{\mcitedefaultmidpunct}
{\mcitedefaultendpunct}{\mcitedefaultseppunct}\relax
\EndOfBibitem
\bibitem{qed1}
Englund D, Faraon A and Fushman I~{\etal} 2007 {\em Nature\/}
  \href{http://dx.doi.org/doi:10.1038/nature06234}{{\bf 440} 857--861}\relax
\mciteBstWouldAddEndPuncttrue
\mciteSetBstMidEndSepPunct{\mcitedefaultmidpunct}
{\mcitedefaultendpunct}{\mcitedefaultseppunct}\relax
\EndOfBibitem
\bibitem{qed2}
Niemczyk T, Deppe F and Huebl H~{\etal} 2010 {\em Nature Physics\/}
  \href{http://dx.doi.org/doi:10.1038/nphys1730}{{\bf 6} 772–776}\relax
\mciteBstWouldAddEndPuncttrue
\mciteSetBstMidEndSepPunct{\mcitedefaultmidpunct}
{\mcitedefaultendpunct}{\mcitedefaultseppunct}\relax
\EndOfBibitem
\bibitem{jj}
Sornborger A~T, Cleland A~N and Geller M~R 2004 {\em Phys. Rev.\/} A
  \href{http://dx.doi.org/10.1103/PhysRevA.70.052315}{{\bf 70}(5) 052315}\relax
\mciteBstWouldAddEndPuncttrue
\mciteSetBstMidEndSepPunct{\mcitedefaultmidpunct}
{\mcitedefaultendpunct}{\mcitedefaultseppunct}\relax
\EndOfBibitem
\bibitem{tj}
Leibfried D, Blatt R and Monroe C~{\etal} 2003 {\em Rev. Mod. Phys.\/}
  \href{http://dx.doi.org/10.1103/RevModPhys.75.281}{{\bf 75}(1)
  281--324}\relax
\mciteBstWouldAddEndPuncttrue
\mciteSetBstMidEndSepPunct{\mcitedefaultmidpunct}
{\mcitedefaultendpunct}{\mcitedefaultseppunct}\relax
\EndOfBibitem
\bibitem{super}
Johansson J, Saito S and Meno T~{\etal} 2006 {\em Phys. Rev. Lett.\/}
  \href{http://dx.doi.org/10.1103/PhysRevLett.96.127006}{{\bf 96}(12)
  127006}\relax
\mciteBstWouldAddEndPuncttrue
\mciteSetBstMidEndSepPunct{\mcitedefaultmidpunct}
{\mcitedefaultendpunct}{\mcitedefaultseppunct}\relax
\EndOfBibitem
\bibitem{semi}
Hennessy K, Badolato A and Winger M~{\etal} 2007 {\em Nature\/}
  \href{http://dx.doi.org/doi:10.1038/nature05586}{{\bf 445} 896--899}\relax
\mciteBstWouldAddEndPuncttrue
\mciteSetBstMidEndSepPunct{\mcitedefaultmidpunct}
{\mcitedefaultendpunct}{\mcitedefaultseppunct}\relax
\EndOfBibitem
\bibitem{braak}
Braak D 2011 {\em Phys. Rev. Lett.\/}
  \href{http://dx.doi.org/10.1103/PhysRevLett.107.100401}{{\bf 107}(10)
  100401}\relax
\mciteBstWouldAddEndPuncttrue
\mciteSetBstMidEndSepPunct{\mcitedefaultmidpunct}
{\mcitedefaultendpunct}{\mcitedefaultseppunct}\relax
\EndOfBibitem
\bibitem{zieg}
Ziegler K 2012 {\em J. Phys. A: Math. Theor.\/}
  \href{http://dx.doi.org/10.1088/1751-8113/45/45/452001}{{\bf 45}
  452001}\relax
\mciteBstWouldAddEndPuncttrue
\mciteSetBstMidEndSepPunct{\mcitedefaultmidpunct}
{\mcitedefaultendpunct}{\mcitedefaultseppunct}\relax
\EndOfBibitem
\bibitem{jc}
Romanelli A 2009 {\em Phys. Rev. \/}A
  \href{http://dx.doi.org/10.1103/PhysRevA.80.014302}{{\bf 80}(1) 014302}\relax
\mciteBstWouldAddEndPuncttrue
\mciteSetBstMidEndSepPunct{\mcitedefaultmidpunct}
{\mcitedefaultendpunct}{\mcitedefaultseppunct}\relax
\EndOfBibitem
\bibitem{gardas4}
Gardas B 2011 {\em J. Phys. A: Math. Theor.\/}
  \href{http://dx.doi.org/10.1088/1751-8113/44/19/195301}{{\bf 44}
  195301}\relax
\mciteBstWouldAddEndPuncttrue
\mciteSetBstMidEndSepPunct{\mcitedefaultmidpunct}
{\mcitedefaultendpunct}{\mcitedefaultseppunct}\relax
\EndOfBibitem
\bibitem{simple}
Chen Q~H, Wang C and He S~{\etal} 2012 {\em Phys. Rev.\/} A
  \href{http://dx.doi.org/10.1103/PhysRevA.86.023822}{{\bf 86}(2) 023822}\relax
\mciteBstWouldAddEndPuncttrue
\mciteSetBstMidEndSepPunct{\mcitedefaultmidpunct}
{\mcitedefaultendpunct}{\mcitedefaultseppunct}\relax
\EndOfBibitem
\bibitem{bom}
Langer H and Tretter C 1998 {\em J. Operator Theory\/} {\bf 39} 339--359\relax
\mciteBstWouldAddEndPuncttrue
\mciteSetBstMidEndSepPunct{\mcitedefaultmidpunct}
{\mcitedefaultendpunct}{\mcitedefaultseppunct}\relax
\EndOfBibitem
\bibitem{ricc_book}
Egoriv A~I 2007 {\em Riccati Equations\/} (Pensoft Publishers, Bulgaria)\relax
\mciteBstWouldAddEndPuncttrue
\mciteSetBstMidEndSepPunct{\mcitedefaultmidpunct}
{\mcitedefaultendpunct}{\mcitedefaultseppunct}\relax
\EndOfBibitem
\bibitem{gardas}
Gardas B 2010 {\em J. Math. Phys.\/}
  \href{http://dx.doi.org/10.1063/1.3442364}{{\bf 51} 062103}\relax
\mciteBstWouldAddEndPuncttrue
\mciteSetBstMidEndSepPunct{\mcitedefaultmidpunct}
{\mcitedefaultendpunct}{\mcitedefaultseppunct}\relax
\EndOfBibitem
\bibitem{gardas2}
Gardas B 2011 {\em J. Math. Phys.\/}
  \href{http://dx.doi.org/10.1063/1.3574889}{{\bf 52} 042104}\relax
\mciteBstWouldAddEndPuncttrue
\mciteSetBstMidEndSepPunct{\mcitedefaultmidpunct}
{\mcitedefaultendpunct}{\mcitedefaultseppunct}\relax
\EndOfBibitem
\bibitem{gardasNote}
Gardas B and Pucha\l{}a Z 2012 {\em J. Math. Phys.\/}
  \href{http://dx.doi.org/10.1063/1.3676309}{{\bf 53} 012106}\relax
\mciteBstWouldAddEndPuncttrue
\mciteSetBstMidEndSepPunct{\mcitedefaultmidpunct}
{\mcitedefaultendpunct}{\mcitedefaultseppunct}\relax
\EndOfBibitem
\bibitem{simon}
Reed M and Simon B 1980 {\em Method of Modern Mathematical Physics\/} (Academic
  Press, London)\relax
\mciteBstWouldAddEndPuncttrue
\mciteSetBstMidEndSepPunct{\mcitedefaultmidpunct}
{\mcitedefaultendpunct}{\mcitedefaultseppunct}\relax
\EndOfBibitem
\bibitem{berezin}
Berezin F~A 1966 {\em Method of Second Quantization\/} vol~24 (Academic Press,
  NY)\relax
\mciteBstWouldAddEndPuncttrue
\mciteSetBstMidEndSepPunct{\mcitedefaultmidpunct}
{\mcitedefaultendpunct}{\mcitedefaultseppunct}\relax
\EndOfBibitem
\bibitem{br}
Bratteli O and Robinson W 1997 {\em Operator Algebras and Quantum Statistical
  Mechanics 2\/} (Springer-Verlag, Berlin)\relax
\mciteBstWouldAddEndPuncttrue
\mciteSetBstMidEndSepPunct{\mcitedefaultmidpunct}
{\mcitedefaultendpunct}{\mcitedefaultseppunct}\relax
\EndOfBibitem
\bibitem{sz0}
Szafraniec F~H 1998 {\em Contemp. Math.\/} {\bf 212} 269--276\relax
\mciteBstWouldAddEndPuncttrue
\mciteSetBstMidEndSepPunct{\mcitedefaultmidpunct}
{\mcitedefaultendpunct}{\mcitedefaultseppunct}\relax
\EndOfBibitem
\bibitem{blank}
Blank J, Exner P and Havli{\v{c}}ek M 2008 {\em Hilbert Space Operators in
  Quantum Physics\/} (Springer, NY)\relax
\mciteBstWouldAddEndPuncttrue
\mciteSetBstMidEndSepPunct{\mcitedefaultmidpunct}
{\mcitedefaultendpunct}{\mcitedefaultseppunct}\relax
\EndOfBibitem
\bibitem{strong}
Albeverio S and Motovilov A~K 2011 {\em Trans. Moscow Math. Soc.\/} {\bf 72}
  45--77\relax
\mciteBstWouldAddEndPuncttrue
\mciteSetBstMidEndSepPunct{\mcitedefaultmidpunct}
{\mcitedefaultendpunct}{\mcitedefaultseppunct}\relax
\EndOfBibitem
\bibitem{main}
Albeverio S, Makarov K~A and Motovilov A~K 2003 {\em Can. J. Math.\/} {\bf
  55}(3) 449--503\relax
\mciteBstWouldAddEndPuncttrue
\mciteSetBstMidEndSepPunct{\mcitedefaultmidpunct}
{\mcitedefaultendpunct}{\mcitedefaultseppunct}\relax
\EndOfBibitem
\bibitem{approx}
He S, Wang C and Chen Q~H~{\etal} 2012 {\em Phys. Rev.\/} A
  \href{http://dx.doi.org/10.1103/PhysRevA.86.033837}{{\bf 86}(3) 033837}\relax
\mciteBstWouldAddEndPuncttrue
\mciteSetBstMidEndSepPunct{\mcitedefaultmidpunct}
{\mcitedefaultendpunct}{\mcitedefaultseppunct}\relax
\EndOfBibitem
\bibitem{srwa}
Albert V~V, Scholes G~D and Brumer P 2011 {\em Phys. Rev.\/} A
  \href{http://dx.doi.org/10.1103/PhysRevA.84.042110}{{\bf 84}(4) 042110}\relax
\mciteBstWouldAddEndPuncttrue
\mciteSetBstMidEndSepPunct{\mcitedefaultmidpunct}
{\mcitedefaultendpunct}{\mcitedefaultseppunct}\relax
\EndOfBibitem
\bibitem{judd}
Judd B~R 1977 {\em J. Chem. Phys\/}
  \href{http://dx.doi.org/10.1063/1.434971}{{\bf 67} 1174--1179}\relax
\mciteBstWouldAddEndPuncttrue
\mciteSetBstMidEndSepPunct{\mcitedefaultmidpunct}
{\mcitedefaultendpunct}{\mcitedefaultseppunct}\relax
\EndOfBibitem
\bibitem{em}
Emary C and Bishop R~F 2002 {\em J. Phys. A: Math. and Gen.\/}
  \href{http://dx.doi.org/doi:10.1088/0305-4470/35/39/307}{{\bf 35} 8231}\relax
\mciteBstWouldAddEndPuncttrue
\mciteSetBstMidEndSepPunct{\mcitedefaultmidpunct}
{\mcitedefaultendpunct}{\mcitedefaultseppunct}\relax
\EndOfBibitem
\bibitem{qes}
{Ko\c c} R, Koca M and T\"u{}t\"u{}nc\"u{}ler H 2002 {\em J. Phys. A: Math. and
  Gen.\/} \href{http://dx.doi.org/doi:10.1088/0305-4470/35/44/311}{{\bf 35}
  9425}\relax
\mciteBstWouldAddEndPuncttrue
\mciteSetBstMidEndSepPunct{\mcitedefaultmidpunct}
{\mcitedefaultendpunct}{\mcitedefaultseppunct}\relax
\EndOfBibitem
\bibitem{reik}
Reik H~G and Doucha M 1986 {\em Phys. Rev. Lett.\/}
  \href{http://dx.doi.org/10.1103/PhysRevLett.57.787}{{\bf 57}(7)
  787--790}\relax
\mciteBstWouldAddEndPuncttrue
\mciteSetBstMidEndSepPunct{\mcitedefaultmidpunct}
{\mcitedefaultendpunct}{\mcitedefaultseppunct}\relax
\EndOfBibitem
\bibitem{gqn}
Gardas B and Dajka J 2012 Initial states of qubit-boson models leading to
  conserved quantities (\textit{Preprint} \eprint{1301.5661}), \it{accepted}\relax
\mciteBstWouldAddEndPuncttrue
\mciteSetBstMidEndSepPunct{\mcitedefaultmidpunct}
{\mcitedefaultendpunct}{\mcitedefaultseppunct}\relax
\EndOfBibitem
\end{mcitethebibliography}
\providecommand{\newblock}{}

\end{document}